\documentclass[a4paper,11pt]{article}

\usepackage[utf8]{inputenc}
\usepackage{fullpage}
\usepackage{graphicx}
\usepackage{amsmath,amssymb,amsthm}
\usepackage{enumerate}
\usepackage[numbers,sort&compress]{natbib}
\usepackage[%
  pdftitle={Disimplicial arcs, transitive vertices, and disimplicial eliminations},%
  pdfauthor={Martiniano Eguía and Francisco J.\ Soulignac},%
  pdfcreator={},%
  pdfsubject={Disimplicial arcs, transitive vertices, and disimplicial eliminations},%
  pdfkeywords={disimplicial arcs, bisimplicial edges of bipartite graphs, disimplicial elimination schemes, bisimplicial elimination schemes, diclique irreducible digraphs, transitive digraphs, dedekind digraphs},%
  colorlinks=true,%
  linkcolor=blue,%
  citecolor=blue,%
  urlcolor=blue]{hyperref}

\newtheorem{theorem}{Theorem}
\newtheorem{lemma}[theorem]{Lemma}
\newtheorem{corollary}[theorem]{Corollary}
\newtheorem{proposition}[theorem]{Proposition}

\title{Disimplicial arcs, transitive vertices, and disimplicial eliminations}
\author{Martiniano Eguía\thanks{Departamento de Computaci\'on, FCEN, Universidad de Buenos Aires, 
Buenos Aires, Argentina.} \and Francisco J.\ Soulignac\thanks{CONICET, Departamento de Computaci\'on, FCEN, Universidad de Buenos Aires, and Departamento de Ciencia y Técnica, Universidad Nacional de Quilmes, Bernal, Argentina.}}

\date{\normalsize\texttt{meguia@dc.uba.ar, francisco.soulignac@unq.edu.ar}}

\newcommand{\SPLIT}{\ensuremath{\mathrm{Split}}}
\newcommand{\JOIN}{\ensuremath{\mathrm{Join}}}
\newcommand{\IN}{\ensuremath{\mathrm{in}}}
\newcommand{\OUT}{\ensuremath{\mathrm{out}}}
\newcommand{\REPG}{\ensuremath{\mathrm{Repr}}}
\newcommand{\TN}{\ensuremath{\theta}}
\newcommand{\REP}{\ensuremath{\mathrm{repr}}}

\begin{document}
\maketitle

\begin{abstract}
 In this article we deal with the problems of finding the disimplicial arcs of a digraph and recognizing some interesting graph classes defined by their existence.  A \emph{diclique} of a digraph is a pair $V \to W$ of sets of vertices such that $v \to w$ is an arc for every $v \in V$ and $w \in W$.  An arc $v \to w$ is \emph{disimplicial} when $N^-(w) \to N^+(v)$ is a diclique.  We show that the problem of finding the disimplicial arcs is equivalent, in terms of time and space complexity, to that of locating the transitive vertices.  As a result, an efficient algorithm to find the bisimplicial edges of bipartite graphs is obtained.  Then, we develop simple algorithms to build disimplicial elimination schemes, which can be used to generate bisimplicial elimination schemes for bipartite graphs.  Finally, we study two classes related to perfect disimplicial elimination digraphs, namely weakly diclique irreducible digraphs and diclique irreducible digraphs.  The former class is associated to finite posets, while the latter corresponds to dedekind complete finite posets.

 \vspace*{.2\baselineskip} {\bf Keywords:} disimplicial arcs, bisimplicial edges of bipartite graphs, disimplicial elimination schemes, bisimplicial elimination schemes, diclique irreducible digraphs, transitive digraphs, dedekind digraphs.
\end{abstract}

\section{Introduction}

Disimplicial arcs are important when Gaussian elimination is performed on a sparse matrix, as they correspond to the entries that preserve zeros when chosen as pivots.  Let $M$ be an $n\times n$ matrix and $G(M)$ be the digraph that has a vertex $r_i$ for each row of $M$ and a vertex $c_j$ for each column of $M$, where $r_i \to c_j$ is an arc of $G(M)$ if and only if $m_{ij} \neq 0$.  The \emph{fill-in} of $m_{ij}$ is the number of zero entries of $M$ that change into a non-zero value when $m_{ij}$ is the next pivot.  To reduce the extra space required to represent $M$, the idea is to pivot with an entry of minimum fill-in.   The extreme case in which $m_{ij}$ has zero fill-in happens when $m_{xy} \neq 0$ for every $x,y$ such that $m_{iy} \neq 0$ and $m_{xj} \neq 0$.  Translated to $G(M)$, the arc $r_i \to c_j$ has ``zero fill-in'' if and only if $r_x \to c_y$ is an arc of $G(M)$ for every $x,y$ such that $r_i \to c_y$ and $r_x \to c_j$ are arcs of $G(M)$.  In graph theoretical terms, the arcs with ``zero fill-in'' are the \emph{disimplicial} arcs of $G(M)$, i.e., the arcs that belong to a unique diclique of $G(M)$.

The discussion above is usually described in terms of bisimplicial edges of bipartite graphs, and not in terms of the disimplicial arcs of digraphs.  We emphasize that these concepts are equivalent for $G(M)$.  Say that a digraph is a \emph{source-sink (ST) graph} when every vertex is either a source or a sink.  Clearly, there are two ST graphs for every bipartite graph $G = (V, W, E)$, depending on whether the edges are oriented from $V$ to $W$ or from $W$ to $V$.  Moreover, there is a one-to-one correspondence between the bisimplicial edges of $G$ and the disimplicial arcs of its orientations.  Thus, it is unimportant whether $G(M)$ is oriented or non-oriented.  There is a reason why we work with digraphs in this manuscript that has to do with the fact that we relate the disimplicial arcs of ST graphs with the vertices of transitive digraphs.  So, in this way we need not describe how the edges of a non-oriented graph should be oriented.

Finding the disimplicial arcs of a digraph $D$ is an interesting and somehow unexplored problem.  It is rather simple to determine if an arc is disimplicial in $O(m)$ time, thus all the disimplicial arcs can be obtained in $O(m^2)$ time and $O(m)$ space.  (We use $n$ and $m$ to denote the number of vertices and arcs of $D$.  Also, and we assume $D$ connected, hence $m \geq n-1$.)  As we shall see in Section~\ref{sec:disimplicial vs transitive}, this problem can be reduced to that of finding the disimplicial arcs of an ST graph $G$.  As it was noted by Bomhoff and Manthey in~\cite{BomhoffMantheyDAM2013}, the twin reduction $G'$ of $G$ can have at most $\tau$ disimplicial arcs, where $\tau < n$ is the number of \emph{thin arcs} of $G'$.  This yields an $O(\tau m)$ time and $O(m)$ space algorithm to find all the disimplicial arcs of $G$.  Bomhoff and Manthey also show that certain random graphs have a constant number of thin arcs, in which case the algorithm takes linear time.  Fast matrix multiplication can also be used to obtain the disimplicial arcs in $O(n^\omega)$ time, but at the expense of $\Theta(n^2)$ space.  This algorithm is, therefore, not convenient for $G$ sparse.

In the process of Gaussian elimination not only the next pivot is important; the whole sequence of pivots is of matter.  Ideally, we would like to use no extra space throughout the algorithm to represent the input matrix $M$.  Thus, no zero entry of $M$ should be changed into a non-zero entry in the entire elimination process.  In~\cite{GolumbicGossJGT1978}, Golumbic~and~Goss observed that this problem corresponds to finding a perfect elimination scheme of $G(M)$.  An \emph{elimination scheme} of a digraph $G$ is a sequence of arcs $S=v_1 \to w_1, \ldots, v_k \to w_k$ such that $v_i \to w_i$ is disimplicial in $G_i = G \setminus \{v_1, w_1, \ldots, v_{i-1}, w_{i-1}\}$, for every $1 \leq i \leq k$.  The sequence $S$ is \emph{maximal} when $G_k$ has no disimplicial arcs, while it is \emph{perfect} when $G_k$ has no edges at all.  Not every digraph admits a perfect elimination scheme; those that do admit it are said to be \emph{perfect elimination}.  In~\cite{GolumbicGossJGT1978} it is proven that every maximal elimination scheme of $G$ is perfect when $G$ is a perfect elimination ST graph.  

The first algorithm to compute a maximal elimination scheme of an ST graph was given by Golumbic and Goss in the aforementioned article.  The algorithm works by iteratively removing the endpoints of a disimplicial arc until no more disimplicial arcs remain.  The complexity of their algorithm is not explicit in~\cite{GolumbicGossJGT1978}; if the disimplicial arcs are searched for as in~\cite{BomhoffMantheyDAM2013}, then $O(\tau nm) = O(n^2m)$ time and $O(m)$ space is required.  Goh and Rotem~\cite{GohRotemIPL1982} propose an $O(n^3)$ time and $O(n^2)$ space algorithm, which was later improved by Bomhoff so as to run in $O(nm)$ time~\cite{Bomhoff2011}.  For the densest cases, the algorithm by Spinrad~\cite{SpinradDAM2004} runs in $O(n^3/\log n)$ time and $O(n^2)$ space.  In~\cite{Bomhoff2011}, Bomhoff shows the most efficient algorithm for the sparse case up to this date, requiring $O(m^2)$ time while consuming $O(m)$ space.

A common restriction of the zero fill-in problem is to ask all the pivots to belong to the diagonal of $M$.  This problem is equivalent to that of finding a perfect elimination scheme whose arcs all belong to some input matching $E$ of $G(M)$.  The matching $E$ represents the arcs that correspond to the diagonal entries of $M$.  Again, this problem can be solved by finding an elimination scheme $S \subseteq E$ such that no arc of $E \setminus S$ is disimplicial in $G(M) \setminus V(S)$~\cite{GolumbicGossJGT1978}.  Rose and Tarjan~\cite{RoseTarjanSJAM1978} devise two algorithms for finding such an elimination scheme of an ST graph, one runs in $O(nm)$ time and space, and the the other requires $O(n^2m)$ time but consumes only $O(m)$ space.  The $O(m^2)$ time algorithm by Bomhoff for finding an unrestricted scheme works in $O(nm)$ time and $O(m)$ space for this case.

In this manuscript we consider two classes related to perfect elimination digraphs, namely diclique irreducible and weakly diclique irreducible graphs.  As far as our knowledge extends, these graph classes were not studied previously.  The motivating question is when does an ST graph $G$ admit a perfect matching $E$ of disimplicial arcs.  For such graphs, any permutation of $E$ is a perfect elimination scheme, thus the pivots of the matrix associated to $G$ can be taken in any order from $E$ with zero fill-in.  How to answer this question efficiently is already known, as it reduces to establishing if the thin arcs form a perfect matching of disimplicial arcs (see~\cite{BomhoffMantheyDAM2013} and Section~\ref{sec:disimplicial vs transitive}).  Nevertheless, the class defined by these graphs has some interesting properties.  Note that, by definition, the arc set of $G$ can be partitioned into a family of dicliques, all of which contain a disimplicial arc.  This resembles the definition of \emph{weakly clique irreducible} graphs~\cite{WangCN2003}, in which every edge should belong to a clique that contains a simplicial edge.  For this reason is that we say a digraph $G$ is \emph{weakly diclique irreducible (WDI)} when every arc of $G$ belongs to a diclique that contains a disimplicial arc.  The word ``weakly'' in the definition of weakly clique irreducible graphs comes from the fact that this is a superclass of the clique irreducible graphs.  A graph is \emph{clique irreducible} when every maximal clique has a simplicial edge~\cite{WallisZhangJCMCC1990}.  By analogy, we define the \emph{diclique irreducible (DI)} digraphs as those digraphs in which every maximal diclique has a disimplicial arc.  

We are mainly interesting on the above problems restricted to sparse digraphs, where sparseness is well distributed.  By this, we mean that we expect each subdigraph to be sparse as well.  The \emph{arboricity} $\alpha$ of a digraph correctly measures this kind of density, as it is the maximum value $e/v$ for a subdigraph with $e$ arcs and $v+1$ vertices~\cite{Nash-WilliamsJLMS1964}.  So, rephrasing, we are mainly interest in the case in which $\alpha \ll n/2$.  Sometimes, however, our algorithms are most efficient when the input digraph is sparse in a stronger sense, as it must have low $h$-index or low maxdegree.  The \emph{$h$-index} is the maximum $\eta$ such that the graph has $\eta$ vertices with degree at least $\eta$, while the \emph{maxdegree} $\Delta$ is the maximum among the degrees of the vertices; it is well known that $\alpha \leq \eta \leq \Delta$ (see e.g.~\cite{LinSoulignacSzwarcfiterTCS2012}).

The article is organized as follows. In Section~\ref{sec:preliminaries} we introduce the terminology used.  In Section~\ref{sec:disimplicial vs transitive} we show two simple operators that transform disimplicial arcs into transitive vertices and back.  As a consequence, finding the disimplicial arcs and finding the transitive vertices are equally hard problems.  In particular, an $O(\min\{\alpha,\tau\}m)$ time and $O(m)$ space algorithm for a digraph with $\tau$ thin arcs is obtained, improving over the algorithm in~\cite{BomhoffMantheyDAM2013}.  This algorithm is optimal unless an $o(\alpha m)$ time algorithm for finding the transitive vertices of a sparse graph is obtained, which is an open problem~\cite{Spinrad2003}.  In Section~\ref{sec:disimplicial elimination} we study the problem of generating maximal elimination schemes.  For the general case we show an algorithm that runs in $O(\min\{\eta\Delta, m\}m)$ time and $O(m)$ space.  The improvement with respect to the algorithm in~\cite{Bomhoff2011} is significant for graphs with $\Delta \ll \sqrt{m}$.  For the case in which all the arcs of the elimination scheme must belong to an input matching, we develop an $O(\alpha m)$ time and $O(m)$ space; which is a major improvement for sparse graphs.  The classes of WDI and DI graphs are studied in Section~\ref{sec:WDI and DI}.  We show that the operators of Section~\ref{sec:disimplicial vs transitive} provide a bijection $f$ between a subfamily of WDI digraphs and finite posets.  When DI digraphs are considered, the range of $f$ are precisely the dedekind complete finite posets, i.e., the finite posets that satisfy the least upper bound property.  With respect to the recognition problems, it can be solved in $O(\alpha m)$ time for WDI digraphs and in $O(nm)$ time for DI digraphs.  Finally, in Section~\ref{sec:further remarks} we translate all the results to bipartite graphs while we provide further remarks.
  
\section{Preliminaries}\label{sec:preliminaries}

A \emph{digraph} is a pair $D = (V(D), E(D))$ where $V(D)$ is finite and $E(D) \subseteq V(D) \times V(D)$; $V(D)$ and $E(D)$ are the \emph{vertex set} and \emph{arc set} of $D$, respectively.  We write $v \to w$ to denote the arc with \emph{endpoints} $v$ and $w$ that \emph{leaves} $v$ and \emph{enters} $w$, regardless of whether $(v,w) \in E(D)$ or not.  Note that our definition allows $D$ to have an arc $v \to v$ for any $v \in V(D)$; in such case, $v$ is a \emph{reflexive} vertex and $v \to v$ is a \emph{loop}.  For $V \subseteq V(D)$, we write $D[V]$ to denote the subdigraph of $D$ \emph{induced} by $V$, and $D \setminus V$ to denote $D[V(D) \setminus V]$. 

For $v \in V(D)$, define $N_D^+(v) = \{w \in V(D) \mid v \to w \in E(D)\}$, $N_D^-(v) = \{w \in V(D) \mid w \to v \in E(D)\}$, and $N_D(v) = N_D^+(v) \cup N_D^-(v)$.  Sets $N_D^+(v)$, $N_D^-(v)$, and $N_D(v)$ are respectively the \emph{out-neighborhood}, \emph{in-neighborhood}, and \emph{neighborhood} of $v$ in $D$, while the members of $N_D^+(v)$, $N_D^-(v)$, and $N_D(v)$ are the \emph{out-neighbors}, \emph{in-neighbors}, and \emph{neighbors} of $v$, respectively.  The \emph{out-degree}, \emph{in-degree}, and \emph{degree} of $v$ are the values $d_D^+(v) = |N_D^+(v)|$, $d_D^-(v) = |N_D^-(v)|$, and $d_D(v) = |N_D(v)|$, respectively.  We omit the subscript from $N$ and $d$ whenever $D$ is clear from context.

For $v \in V(D)$, we say that $v$ is a \emph{source} when $d^-(v) = 0$,  $v$ is a \emph{sink} when $d^+(v) = 0$, and $v$ is \emph{transitive} when $x \to y$ for every $x \in N^-(v)$ and $y \in N^+(v)$.  A digraph is a \emph{source-sink (ST) graph} when it contains only source and sink vertices, while it is \emph{transitive} when it contains only transitive vertices.  A digraph is \emph{simple} when it has no loops, while it is \emph{reflexive} when every vertex is reflexive.  The \emph{reflexive closure} of $D$ is the digraph obtained by adding all the missing loops to $D$ so as to make each vertex reflexive, i.e., the reflexive closure of $D$ is $(V(D), E(D) \cup \{(v,v) \mid v \in V(D)\})$.  An \emph{oriented graph} is a digraph such that $v \to w \in E(D)$ and $w \to v \in E(D)$ only if $v = w$.  An \emph{order graph} is an oriented graph that is simultaneously reflexive and transitive.  Let $\leq$ be the relation on $V(D)$ such that $v\leq w$ if and only if $v \to w \in V(D)$.  Note that $\leq$ is reflexive (resp.\ antisymmetric, transitive) precisely when $D$ is reflexive (resp.\ oriented, transitive).  Thus, $D$ is an order graph if and only if $(V(D), \leq)$ is a finite poset.

For $v \in V(D)$, we write $H_D^+(v) = \{w \in N_D^+(v) \mid d^+(v) \leq d^-(w)\}$ and $H_D^-(v) = \{w \in N_D^-(v) \mid d^-(v) \leq d^+(w)\}$.  In other words, $H_D^+(v)$ has the out-neighbors of $v$ whose in-degree is greater than or equal to the out-degree of $v$, while $H_D^-(v)$ has the out-neighbors of $v$ with in-degree at least $d^-(v)$.  Note that either $v \in H^-(w)$ or $w \in H^+(v)$ for every arc $v \to w \in E(D)$, thus all the arcs of $D$ get visited when all the $H$ sets are traversed.  The values $|H^+_D(v)|$, $|H^-_D(v)|$ are denoted by $h^+_D(v)$ and $h^-_D(v)$, while $h_D(v) = \max\{h_D^+(v), h_D^-(v)\}$.  Again, we omit the subscript $D$ when no ambiguities arise.

We write $n_D$, $m_D$, and $\Delta_D$ to denote the values $|V(D)|$, $|E(D)|$, and $\max_{v \in V(D)}\{d(v)\}$, respectively.  The arboricity and $h$-index are values that measure how dense is a digraph.  We use a non-standard definition of arboricity given by the equivalence in~\cite{Nash-WilliamsJLMS1964}, i.e., the \emph{arboricity} $\alpha_D$ of $D$ is the maximum $e/v$ such that $D$ has a subdigraph with $e$ arcs and $v+1$ vertices.  The \emph{$h$-index} is the value $\eta_D$ such that $D$ has $\eta_D$ vertices with degree at least $\eta_D$.  It is well known that $\alpha_D \leq \eta_D \leq \min\{\Delta, \sqrt{2m_D}\}$, while $h(v) \leq \eta_D$ for every $v \in V(D)$~\cite{ChibaNishizekiSJC1985,LinSoulignacSzwarcfiterTCS2012}.  The time required to multiply two $n\times n$ matrices is denoted by $O(n^\omega)$; up to this date  $2 \leq \omega \leq 2.3729$~\cite{Williams2012}.  As before, we omit the subscripts $D$ whenever possible.  Also, we assume $m > n$ for all the problems considered with no loss of generality.

Two arcs of $D$ are \emph{independent} when they have no common endpoints.  A \emph{matching} is a set $M$ of pairwise independent arcs.  Sometimes we deal with $M$ as if it were the subgraph of $D$ with vertex set $\{v,w \mid v\to w\in M\}$ and arc set $M$.  Thus, we write $V(M)$ to denote the set of vertices entering or leaving an arc of $M$, or we talk about the unique neighbor of $v$ in $M$, etc.  A matching is \emph{perfect} when $V(M) = V(G)$.  

A \emph{diclique} of $D$ is an ordered pair $(V, W) \subseteq V(D) \times V(D)$ such that $v \to w \in E(D)$ for every $v \in V$ and $w \in W$ (note that every vertex in $V \cap W$ is reflexive).  For the sake of notation, we write $V \to W$ to refer to $(V, W)$, regardless of whether $(V, W)$ is a diclique of not.  The term \emph{diclique} is also used to denote the subdigraph $B$ of $D$ with vertex set $V \cup W$ and arc set $\{v \to w \mid v \in V, w \in W\}$; note that $B$ needs not be an induced subdigraph of $D$.  Thus, for instance, we can talk about the arcs of the diclique $B$.  A diclique $V \to W$ of $D$ is \emph{maximal} when $D$ has no diclique $V \cup V' \to W \cup W'$ for $\emptyset \subset V' \cup W' \subseteq V(D)$.  An arc $v\to w \in E(D)$ is \emph{disimplicial} when $B = N(w) \to N(v)$ is a diclique of $D$; note that $B$ is the unique maximal diclique of $D$ that contains $v \to w$.  In such case, the diclique $B$ is said to be \emph{reduced}, i.e., $B$ is \emph{reduced} when it is maximal and it contains a disimplicial arc.

\section{Disimplicial arcs versus transitive vertices}
\label{sec:disimplicial vs transitive}

By definition, a reflexive vertex $v$ is transitive if and only if $v \to v$ is a disimplicial arc.  Hence, we can find out if a digraph $D$ is transitive by looking if all the loops of its reflexive closure $D^*$ are disimplicial.  This result can be easily strengthen so as to make $D^*$ an ST graph.  

For any digraph $D$, define $\SPLIT(D)$ to be the digraph $G$ that has a vertex $\OUT(v)$ for each non-sink vertex $v$, and a vertex $\IN(w)$ for each non-source vertex $w$, where $\OUT(v) \to \IN(w) \in E(G)$ if and only if $v \to w \in E(D)$, for every $v,w \in V(D)$ (see Figure~\ref{fig:split+join}). Clearly, $\OUT(v)$ and $\IN(w)$ are source and sink vertices, resepctively, hence $G$ is an ST graph.  Moreover, the dicliques of $D$ are ``preserved'' into $G$ as in the next proposition.

\begin{figure}
 \centering
 \includegraphics{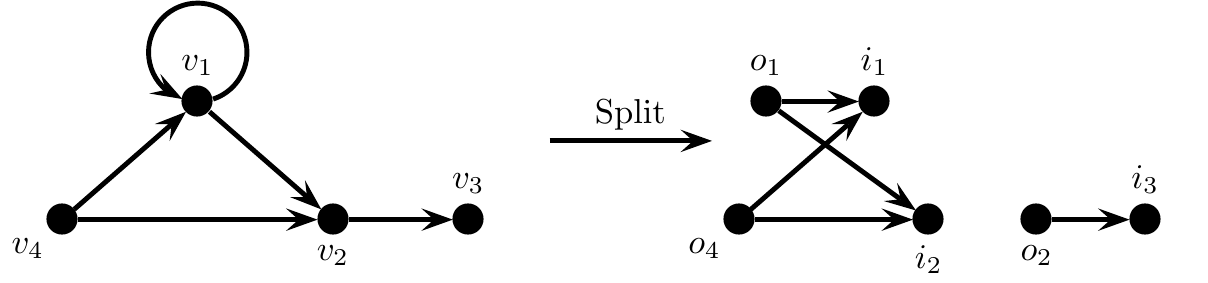}
 
 \includegraphics{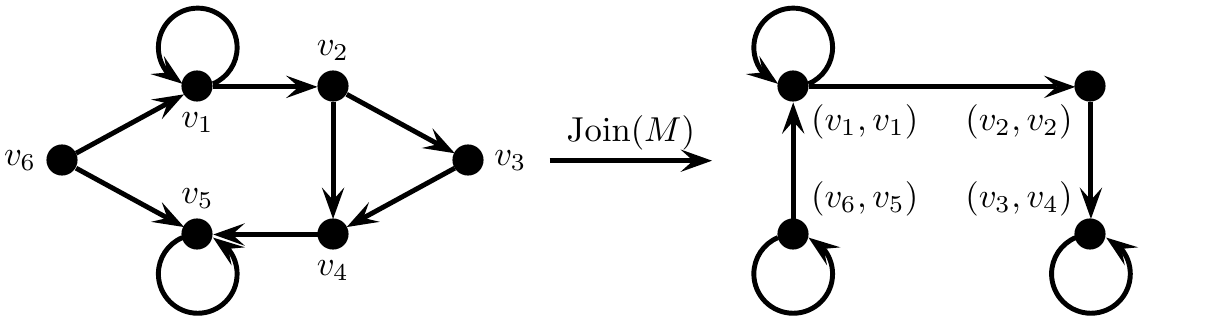}
 \caption{Examples of the operations $\SPLIT(D)$ and $\JOIN(G, M)$ for $M = \{v_1 \to v_1, v_3 \to v_4, v_6 \to v_5\}$.  For the sake of exposition, we write $i_x$ and $o_x$ to denote the vertices $\OUT(v_x)$ and $\IN(v_x)$ of $\SPLIT(D)$, respectively.  Note that $\OUT(v) \to \IN(v)$ is an arc of $\SPLIT(D)$ if only if $v$ is reflexive, while $(v, w)$ is reflexive in $\JOIN(G, M)$ if and only if either $v \to w \in M$ or $v = w$ is reflexive in $G$.}\label{fig:split+join}
\end{figure}

\begin{proposition}\label{prop:SPLIT dicliques}
  Let $D$ be a digraph.  Then, $V \to W$ is a diclique of $D$ if and only if\/ $\OUT(V) \to \IN(W)$ is a diclique of\/ $\SPLIT(D)$, where\/ $\OUT(V) = \{\OUT(v)\}_{v \in V}$ and\/ $\IN(W) = \{\IN(w)\}_{w \in W}$.
\end{proposition}

\begin{corollary}\label{cor:SPLIT disimplicials}
  Let $D$ be a digraph.  Then, $v \to w$ is a disimplicial arc of $D$ if and only if\/ $\OUT(v) \to in(w)$ is a disimplicial arc of\/ $\SPLIT(D)$.
\end{corollary}

So, as anticipated, we can find out whether $D$ is transitive or not by computing the disimplicial arcs of $\SPLIT(D^*)$.  Since $\SPLIT(D^*)$ can be computed in linear time when $D$ is provided as input, we conclude that finding the disimplicial arcs of an ST graph is harder than testing if a digraph is transitive.

\begin{theorem}\label{thm:transitive is disimplicial}
  A digraph $D$ is transitive if and only if all the arcs in the matching\/ $\{\OUT(v) \to \IN(v) \mid v \in V(D^*)\}$ of\/ $\SPLIT(D^*)$ are disimplicial, where $D^*$ is the reflexive closure of $D$.  
\end{theorem}

For the rest of this section, we discuss how to find disimplicial arcs by computing transitive vertices.  The idea is to revert, as much as possible, the effects of $\SPLIT$.  For any matching $M$ of a digraph $G$, define $\JOIN(G,M)$ to be the digraph $D$ that has a vertex $(v,v)$ for each $v \in V(G) \setminus V(M)$, and a vertex $(v,w)$ for each $v \to w \in M$, where $(v,w) \to (x,y) \in E(D)$ if and only if $v \to y \in E(G)$ (see Figure~\ref{fig:split+join}).  The restricted duality between the $\SPLIT$ and $\JOIN$ operators is given in the next lemmas.

\begin{lemma}\label{lem:bip then join}
  If $D$ is a reflexive digraph, then $D$ is isomorphic to\/ $\JOIN(\SPLIT(D), \{\OUT(v) \to \IN(v) \mid v \in V(D)\})$.
\end{lemma}

\begin{proof}
  Note that $M = \{\OUT(v) \to \IN(v) \mid v \in V(D)\}$ is a matching of $D$ because $D$ is reflexive, hence $H = \JOIN(G, M)$ is well defined for $G = \SPLIT(D)$.  Let $f\colon V(D) \to V(H)$ be the function such that $f(v) = (\IN(v), \OUT(v))$ (see Figure~\ref{fig:bip then join}).   By definition of $\SPLIT$, $v \to w \in E(D)$ if and only if $\OUT(v) \to \IN(w) \in E(G)$, for every $v,w \in V(D)$.  Similarly, by the definition of $\JOIN$, $\OUT(v) \to \IN(w) \in E(G)$ if and only if $(\OUT(v), \IN(v)) \to (\OUT(w), \IN(w)) \in E(H)$.  That is, $v \to w \in E(D)$ if and only if $f(v) \to f(w) \in E(H)$.    
\end{proof}

\begin{figure}
 \centering
 \includegraphics{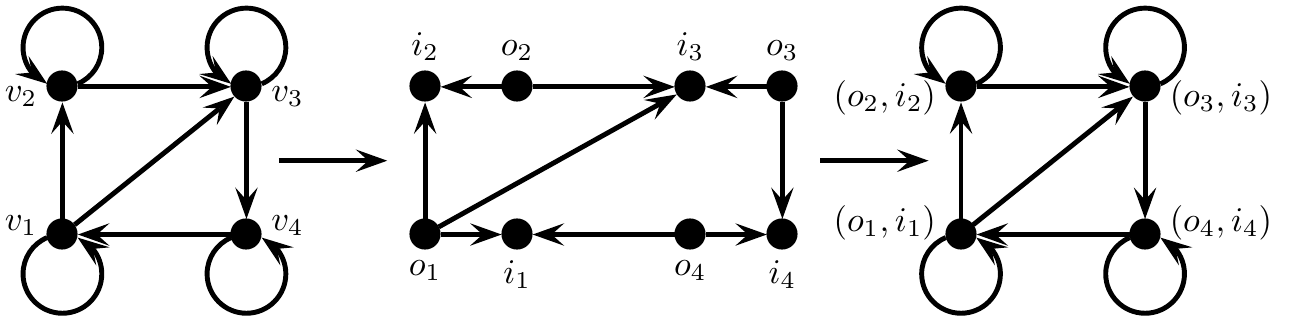}
 \caption{From left to right: $D$, $G = \SPLIT(D)$, and $H = \JOIN(G, M)$ for $M = \{\OUT(V) \to \IN(v) \mid v \in V(D)\}$.   Again, we write $i_x$ and $o_x$ to denote the vertices $\OUT(v_x)$ and $\IN(v_x)$ of $G$, respectively. Note that the function $f$ of Lemma~\ref{lem:bip then join} is an isomorphism between $D$ and $H$.}\label{fig:bip then join}
\end{figure}

\begin{lemma}\label{lem:join then bip}
  If $M$ is a perfect matching of an ST graph $G$, then $G$ is isomorphic to\/ $\SPLIT(\penalty0\JOIN(\penalty0G, M))$.
\end{lemma}

\begin{proof}
  The proof is analogous to that of Lemma~\ref{lem:bip then join}.  This time, take $H = \SPLIT(\JOIN(G, M))$ and $m(v)$ be the neighbor of $v$ in $M$, and observe that $f\colon V(G) \to V(H)$ is an isomorphism when $f(v) = \IN((v, m(v)))$ for every sink vertex $v$ and $f(v) = \OUT((m(v), v))$ for every source vertex $v$.
\end{proof}

Despite Lemma~\ref{lem:join then bip} requires an ST graph $G$ with a perfect matching $M$, the $\JOIN$ operator can be applied to any digraph and any matching.  The final result is always the same, though; the disimplicial arcs of $M$ get transformed into transitive vertices.

\begin{theorem}\label{thm:disimplicial is transitive}
  Let $M$ be a matching of a digraph $G$, and $v \to w \in M$.  Then, $v \to w$ is disimplicial in $G$ if and only if $(v,w)$ is a transitive vertex of\/ $\JOIN(G, M)$.
\end{theorem}

\begin{proof}
  Let $D = \JOIN(G, M)$ and observe that $(v,w) \in V(D)$.  By definition, $(a,b) \to (x,y) \in E(D)$ if and only if $a \to y \in E(G)$, for every $a,b,x,y \in V(G)$.  Then, $(a,b) \to (x,y) \in E(D)$ for every pair $(a,b), (x,y) \in V(D)$ such that $(a, b) \to (v,w) \in E(D)$ and $(v,w) \to (x,y) \in E(D)$ if and only if $a \to y \in E(G)$ for every pair $a, y \in V(G)$ such that $a \to w \in E(G)$ and $v \to y \in E(G)$.  That is, $(v,w)$ is transitive in $D$ if and only if $v \to w$ is disimplicial in $G$.
\end{proof}

Theorem~\ref{thm:disimplicial is transitive} gives us a method for testing if an arc $v \to w$ is disimplicial: check if $(v,w)$ is transitive in $D = \JOIN(G, \{v \to w\})$.  Since $D$ can be computed in $O(d_G(v) + d_G(w))$ time when $G$ and $v \to w$ are given as input, we conclude that querying if an arc is disimplicial is equally hard as determining if a vertex is transitive.  We remark that testing if $(v,w) \in V(D)$ is transitive and checking if $v \to w \in E(G)$ is disimplicial are both solvable in $O(m)$ time.

Theorem~\ref{thm:disimplicial is transitive} can also be used to find all the disimplicial arcs of $G$ when an adequate matching is provided.  For the sake of simplicity, we restrict ourselves to ST graphs, by Proposition~\ref{prop:SPLIT dicliques}.  Moreover, we find it convenient to eliminate twin vertices.  Two vertices $v,w$ of an ST graph $G$ are \emph{twins} when $N(v) = N(w)$, while $G$ is \emph{twin-free} when it contains no pair of twins.  A \emph{twin block} is a maximal set of twin vertices; note that $V(G)$ admits a unique partition into twin blocks.  We assume the existence of a function $\REP_G$ that, given a block $B$, returns a vertex of $B$, and we write $\REP_G(v) = \REP_G(B)$ for every $v \in B$.  For the sake of notation, we omit the subscript $G$ from $\REP$ when no ambiguities arise.  The \emph{twin reduction} of $G$ is the subdigraph $\REPG(G)$ of $G$ induced by $\{\REP(B) \mid B \text{ is a block of } G\}$.  The twin reduction of $G$ contains all the information about the disimplicial arcs of $G$, as in the next proposition.

\begin{proposition}\label{prop:twin reduction}
  An arc $v \to w$ of an ST graph $G$ is disimplicial if and only if\/ $\REP(v) \to \REP(w)$ is disimplicial in $\REPG(G)$. 
\end{proposition}

We are now ready to state what an adequate matching looks like.  For each $v \in V(G)$, define the \emph{thin neighbor} $\TN(v)$ of $v$ to be the (unique) vertex $w \in N(v)$ such that $d(w) < d(z)$ for every $z \in N(v) \setminus \{w\}$; if such a vertex does not exist, then $\TN(v)$ is some undefined vertex.  Say that $v \to w$ is a \emph{thin} arc when $v = \TN(w)$ and $w = \TN(v)$.  For the sake of notation, we write $\JOIN(G)$ to denote $\JOIN(G, M)$ where $M$ is the set of thin arcs of $G$; note that $\JOIN(G)$ is well defined because $M$ is a matching. The following easy-to-prove lemma is as fundamental for us as it is for the algorithm in~\cite{BomhoffMantheyDAM2013}.

\begin{lemma}[see e.g.~\cite{BomhoffMantheyDAM2013}]\label{lem:disimplicial is thin}
  All the disimplicial arcs of a twin-free ST graph are thin. 
\end{lemma}

The algorithm to compute the disimplicial arcs of an ST graph works in two phases.  In the first phase, all the disimplicial arcs of $H = \REPG(G)$ are obtained by querying which of the vertices of $\JOIN(H)$ are transitive.  In the second phase, each $v \to w \in E(G)$ is tested to be disimplicial by querying if $\REP(v) \to \REP(w)$ is disimplicial in $H$.   The algorithm is correct by Theorem~\ref{thm:disimplicial is transitive}, Proposition~\ref{prop:twin reduction}, and Lemma~\ref{lem:disimplicial is thin}.  

\begin{theorem}\label{thm:disimplicial algorithm}
  An arc $v \to w$ of a digraph $G$ is disimplicial if and only if\/ $(\REP(\OUT(v)), \REP(\IN(w)))$ is transitive in\/ $\JOIN(\REPG(\SPLIT(G)))$.
\end{theorem}

Since $\SPLIT$, $\JOIN$, and $\REPG$ can be computed in linear time, we conclude that listing the disimplicial arcs and finding the transitive vertices are equally hard problems. Up to these date, the best algorithms for computing the transitive vertices of $D = \JOIN(H)$ take $O(\alpha_D m_D)$ time and $O(m_D)$ space or $O(n_D^\omega)$ time and $O(n_D^2)$ space. Since $\alpha_D = O(\alpha_G)$, $n_D = O(n_G)$, and $m_D = O(m_G)$, we conclude that the disimplicial arcs of a digraph $G$ can be obtained in either $O(\alpha_G m_G)$ time and $O(m_G)$ space or $O(n_G^\omega)$ time and $O(n_G^2)$ space.

\section{Disimplicial eliminations}
\label{sec:disimplicial elimination}

The present section is devoted to the problems of finding disimplicial elimination sequences.  Before doing so, we review the $h$-digraph structure as it is required by our algorithms.

The \emph{$h$-graph structure} was introduced in~\cite{LinSoulignacSzwarcfiterTCS2012} with dynamic algorithms in mind.  It proved to be well suited for some vertex elimination problems, particularly those in which the conditions for removing a vertex are local to its neighborhood.  The \emph{$h$-digraph structure} is the cousin of $h$-graphs for digraphs, and it was superficially described in~\cite{LinSoulignacSzwarcfiterTCS2012}.  Let $D$ be a digraph and $\{\bullet, \circ\} = \{+,-\}$.  In short, the $h$-digraph structure maintains $3$ values for $\bullet$ and each $v \in V(D)$, namely $d^\bullet(v)$, $\mathcal{N}^\bullet(v)$, and $H^\bullet(v)$, where $\mathcal{N}^\bullet$ is an ordered list of the nonempty sets $N^\bullet(v, i)$ = $\{z \in N^\bullet(v) \mid d^\circ(z) = i\}$ with $i < d^\bullet(v)$.  Recall that $H^\bullet = \{z \in N^\bullet(v) \mid d^\circ(z) \geq d^\bullet(v)\}$.    The data structure also keeps track of several pointers that allow efficient access to the different incarnations of a vertex in the structure (see~\cite{LinSoulignacSzwarcfiterTCS2012}).  At all, no more than $O(m)$ bits are consumed.

Table~\ref{tab:h-operations} describes the operations supported by the $h$-digraph structure that are of interest for our purposes.  All of them, but {\tt MinN}, where described in~\cite{LinSoulignacSzwarcfiterTCS2012} for graphs, though their translation to digraphs is direct.  For the implementation of {\tt MinN}, two cases are considered to obtain the desired output $L$.  If $\mathcal{N}^\bullet = \emptyset$, then $d^\bullet(v) \leq d^\circ(w)$ for every $w \in N^\bullet(v)$, thus $L \subseteq H^\bullet(v)$; otherwise, $L$ is equal to the first set in $\mathcal{N}^\bullet(v)$.  The time required for this operation is, therefore, $O(h(v))$.  

\begin{table}
  \centering
  \begin{tabular}{llcc}
    \hline
    Operation & Description & \multicolumn{2}{c}{Complexity}\\
              &             & one & all \\\hline
    {\tt Initialize}($D$) & creates the $h$-graph structure of $D$ & - & $O(\alpha m)$\\
    {\tt Remove}($v, D$) & removes $v \in V(D)$ from $D$  & $O(dh)$ & $O(\alpha m)$ \\
    {\tt N'}($v, D, \bullet$) & returns $\{w \to z \in E(D) \mid w, z \in N^\bullet(v)\}$  & $O(dh)$ & $O(\alpha m)$ \\
    {\tt MinN}($v, D, \bullet$) & returns $\{w \in N^\bullet(v) \mid d^\circ(w) \leq d^\circ(z) \text{ for } z \in N^\bullet(v)\}$ & $O(h)$ & -\\
    {\tt d}($v, D, \bullet$) & returns $d^\bullet(v)$ & $O(1)$ & - \\
    \hline
  \end{tabular}
  \caption{Some operations supported by the $h$-digraph data structure.  The complexity column ``one'' indicates the time required by one invocation of the operation, while the complexity column ``all'' indicates the time required when the operation is applied $O(1)$ times to all the vertices in the digraph.  Here $h = h(v)$, $d = d(v)$, $\alpha = \alpha_D$ and $m = m_D$, $\bullet$ must belong to $\{+,-\}$, and $\circ$ is the opposite of $\bullet$.}\label{tab:h-operations}
\end{table}

\subsection{General disimplicial eliminations}

A sequence of arcs $S = v_1 \to w_1, \ldots, v_k \to w_k$ is a \emph{disimplicial elimination} of a digraph $G$ when $v_i \to w_i$ is disimplicial in $G_i = G \setminus \{v_1, w_1, \ldots, v_{i-1}, w_{i-1}\}$ for every $1 \leq i \leq k$; $S$ is \emph{maximal} when $G_{k+1}$ has no disimplicial arcs.  For convenience, we write $V(S)$ to denote the set of vertices of $S$.  

The algorithm to compute a maximal disimplicial elimination works in an iterative manner from an input digraph $G = G_1$.  At iteration $i$, the algorithm finds a disimplicial elimination $S_i$ of $G_{i}$ by taking any maximal matching of disimplicial arcs of $G_i$.  By maximal, we mean that either $v \in V(S_i)$ or $w \in V(S_i)$ for every disimplicial arc $v \to w$ of $G_i$.  Then, the algorithm updates $G_{i}$ into $G_{i+1} = G_{i} \setminus V(S_i)$ for the iteration $i+1$.  The algorithm stops with output $S = S_1, \ldots, S_{i-1}$ when $S_i = \emptyset$.  

For the sake of notation, in the rest of this section we write $P_i$ to denote each parameter $P$ on $G_i$ instead of using $P_{G_i}$; thus, we write $N_i(v)$ to denote $N_{G_i}(v)$, $\Delta_i$ to denote $\Delta_{G_i}$, and so on.  When no subscript is wrote, the parameter on $G$ should be understood; e.g., $N(v) = N_G(v)$, $\Delta = \Delta_G$, etc. 

The main idea of the algorithm is to compute $S_i$, for $i > 1$, by looking only at the arcs leaving or entering $V(S_{i-1})$.  Of all such arcs, we are interested in those with ``low degree'', which are the analogous of thin arcs for those digraphs that can contain twins (see Proposition~\ref{prop:disimplicial in L} below).   Let $V_{\OUT} = \{v \in V(G_i) \mid v \to y \text{ for } y \in V(S_{i-1})\}$ and $V_{\IN} = \{w \in V(G_i) \mid x \to w \text{ for } x \in V(S_{i-1})\}$, i.e., $V_{\OUT}$ and $V_{\IN}$ are the set of vertices of $G_i$ that have an out and in neighbor that was removed from $G_{i-1}$, respectively. For each $v \in V_{\OUT}$ (resp.\ $V_{\IN}$), let $L(v)$ be the set of out-neighbors (resp.\ in-neighbors) of $v$ with minimum in-degree (resp.\ out-degree) in $G_i$.  To compute $S_i$, the algorithm first initializes $S_i := \emptyset$ and then it traverses each vertex $v \in V_{\OUT} \cup V_{\IN}$.  For $v \in V_{\OUT}$ (resp.\ $v \in V_{\IN}$), the algorithm evaluates whether $v \to \ell$ (resp.\ $\ell \to v$) is disimplicial for any $\ell \in L(v)$.  If affirmative and $L(v) \setminus V(S_i) \neq\emptyset$, then $v \to w$ (resp.\ $w \to v$) is inserted into $S_i$ for any $w \in L(v) \setminus V(S_i)$.  (Note that $w$ needs not be equal to $\ell$; this happens when $x \to \ell$ or $\ell \to x$ was previously inserted into $S_i$ for some $x \in V(G_i)$.) If negative or $L(v) \subseteq V(S_i)$, then $v$ is ignored.  By invariant, $S_i$ is a matching of $G_i$.  Moreover $S_i$ contains only disimplicial arcs, as it follows from the following generalization of Lemma~\ref{lem:disimplicial is thin}.  

\begin{proposition}\label{prop:disimplicial in L}
  Let $v \in V_{\OUT} \cup V_{\IN}$ be an endpoint of some disimplicial arc of $G_i$.  Then, $v \to w$ (resp.\ $w \to v$) is disimplicial in $G_i$ if and only if $w \in L(v)$.
\end{proposition}

The next proposition shows that, as required, $S_i$ is indeed maximal.  That is, the algorithm to compute $S_i$ is correct. 

\begin{proposition}
 If $v \to w$ is a disimplicial arc of $G_i$, then either $v \in V(S_i)$ or $w \in V(S_i)$.
\end{proposition}

\begin{proof}
  Observe that $v \to w$ is not disimplicial in $G_{i-1}$, since otherwise either $v$ or $w$ would have been removed in the update from $G_{i-1}$ to $G_i$, by the maximality of $S_{i-1}$.  Hence, there exist $x,y \in V(G_{i-1})$ such that $y \in N^+(v)$, $x \in N^-(w)$ and $x \to y \not\in E(G)$.  Since $v\to w$ is disimplicial in $G_{i}$, then either $x$ or $y$ does not belong to $G_i$.  In the former case $x \in V(S_{i-1})$ and $w \in V_{\IN}$, while in the latter case $y \in V(S_{i-1})$ and $v \in V_{\OUT}$.  Both cases are analogous, so suppose $v \in V_{\OUT}$.  By Proposition~\ref{prop:disimplicial in L}, $w \in L(v)$, while $v \to \ell$ is disimplicial for every $\ell \in L(v)$.  Consequently, $v$ is ignored by the algorithm (i.e., $v \not\in V(S_i)$) only if $w \in L(v) \subseteq V(S_i)$.
\end{proof}

Each time an arc $v \to w$ is evaluated to be disimplicial, the algorithm works as follows.  First, the vertices in $N_i^+(v) \cup N_i^-(w)$ are marked, and a variable $e$ is initialized to $0$.  The purpose of $e$ is to count the number of arcs that leave a vertex in $N_i^-(w)$ to enter a vertex in $N_i^+(v)$.  To compute $e$, each $x \in H_i^-(y)$ is traversed, for every $y \in N_i^+(v)$. If $x$ is marked, then $x \in N_i^-(w)$ and $y \in N_i^+(v)$, thus $e$ is increased by $1$; otherwise $x \not\in N_i^-(w)$, thus $e$ remains unchanged.  The arc $x \to y$ is also marked so as to avoid counting it again.  When the execution for $N_i^+(v)$ is done, the algorithm proceeds to traverse each $y \in H_i^+(x)$, for every $x \in N_i^-(w)$, increasing $e$ by $1$ when $y$ is marked and $x \to y$ is not.  At the end, all the marks are cleared.  Clearly, $e$ counts the number of arcs of $G_i$ leaving $N_i^-(w)$ and entering $N_i^+(v)$ as each arc $x \to y$ with $x \in N_i^-(w)$ and $y \in N_i^+(v)$ is traversed at least once.  Thus $v \to w$ is disimplicial if and only if $e = d_i^+(v)d_i^-(w)$.

The algorithm implements $G_i$ with the $h$-digraph structure.  To compute $S_i$, the vertices in $V = V_{\OUT} \cup V_{\IN}$ need to be traversed; recall that, by definition, $V = \bigcup_{y \in S_{i-1}}N_i(y)$.  For each traversed $v \in V$, a vertex $\ell \in L(v)$ needs to be located; this costs $O(h_i(v))$ time if the first vertex given by \texttt{MinN} is taken. Following, $v \to \ell$ (or $\ell \to v$) is queried to be disimplicial.  For this, the vertices in $N_i^+(v) \cup N_i^-(\ell)$ are first marked in $O(d_i(v) + d_i(\ell))$, and then $e$ is computed in $O\left(\sum_{z \in N_i(v) \cup N_i(\ell)}h_i(z)\right) = O(\Delta_i \eta_i)$ time.  Moreover, note that every arc is traversed $O(1)$ times, thus $O(\min\{m_i, \Delta_i \eta_i\})$ in actually spent to check if $v \to \ell$ is disimplicial.  When $v \to \ell$ (or $\ell \to v$) is disimplicial, \texttt{MinN} is invoked to obtain $L(v)$, which is then traversed so as to locate the arc $v \to w$ (or $w \to v$) to be inserted into $S_i$.  Note that every vertex $z \in L(v)$ that is traversed while looking for $w$ belongs to $V(S_i)$ at the end of step $i$.  Also, $z$ will be evaluated no more than $O(d_i(z))$ times, once for each $v \in N_i(z)$ such that $L(v)$ is considered.  Thus, all the required traversals to the sets $\{L(v) \mid v \in V\}$ consume $O\left(\sum_{z \in V(S_i)}{d_i(z)}\right)$ time. Summing up, the time required to compute $S_i$ is
\begin{align*}
&O\left(
    \sum_{y \in S_{i-1}}\left(
        \sum_{v \in N(y)}(
           h_i(v) + \min\{m_i, \Delta_i \eta_i\}
        )
     \right) + 
     \sum_{z \in V(S_i)}d_i(z)
\right)
  =\\
&O\left(
  \min\{m, \Delta \eta\}\sum_{y\in S_{i-1}}d(y) + 
  \sum_{z \in V(S_i)}d(z)
\right)
\end{align*}

Before the algorithm starts, $G_1$ is initialized with an invocation to \texttt{Initialize} at the cost of $O(\alpha m)$ time.  Similarly, after each iteration, $V(S_i)$ is removed from $G_i$ using the operation \texttt{Remove}.  Note that each vertex is removed exactly once, hence $O(\alpha m)$ time is totally consumed.   Let $k$ be the number of iterations required by the algorithm and $S$ be the output disimplicial elimination.  Since $S_1$ can be computed in $O(\alpha m)$ time and $\bigcup_{i=1}^k S_i = S$ is a matching, we obtain that the total time required by the algorithm is
\begin{align*}
&O\left(
  \alpha m + 
  \sum_{i=2}^k\left(
      \min\{m, \Delta H\}\sum_{y\in S_{i-1}}d(y) + \sum_{z \in V(S_i)}d(z)
  \right)
\right) =\\ 
&O\left(\alpha m + \min\{m, \Delta H\}\sum_{y \in V(S)}d(y) + \sum_{z \in V(S)}d(z)\right) = O(m\min\{m, \Delta H\})
\end{align*}
Since the $h$-digraph structure uses $O(m)$ bits, the space complexity is linear.

\subsection{Disimplicial $M$-eliminations}

We now consider the restricted problem of finding a maximal disimplicial $M$-elimination of a digraph $G$, when an input matching $M$ is given.  A \emph{disimplicial $M$-elimination} is just a disimplicial elimination $S$ of $G$ included in $M$; $S$ is \emph{maximal} when no arc of $M \setminus S$ is disimplicial in $G \setminus V(S)$.  

This time, the idea is to take advantage of the relation between disimplicial arcs and transitive vertices.  Say that a sequence $v_1, \ldots, v_k$ is a \emph{transitive $V$-elimination} of a digraph $D$, for $V \subseteq V(D)$, when $v_i$ is transitive in $D \setminus \{v_1, \ldots, v_{i-1}\}$, for every $1 \leq i \leq k$.  Suppose $S = v_1 \to w_1 , \ldots, v_k \to w_k \subseteq M$ and let $G_1 = G$ and $M_1 = M$.  For $1 \leq i \leq k$, define
\begin{itemize}
\item $G_{i+1} = G \setminus \{v_i, w_i\}$,
\item $M_{i+1} = M_{i} \setminus \{v_i \to w_i\}$,
\item $S_i = v_1 \to w_1 , \ldots, v_i \to w_i$,
\item $D_i = \JOIN(G_i, M_i)$, 
\item $V_i = \{(v, w) \mid v \to w \in M_i\}$, and
\item $T_i = (v_1, w_1) , \ldots, (v_i, w_i)$
\end{itemize}
By definition, $D_i$ has a vertex $(v,w)$ for each $v \to w \in M_i$ and a vertex $(v,v)$ for each $v \in G_i \setminus V(M_i)$ where $(v, w) \to (x,y)$ is an arc of $D_i$ if and only if $v \to y$.  It is not hard to see, then, that $D_{i+1} = \JOIN(G_{i+1}, M_{i+1}) = \JOIN(G_i, M_i) \setminus \{(v_i, w_i)\} = D_i \setminus \{(v_i, w_i)\}$.  Moreover, by Theorem~\ref{thm:disimplicial is transitive}, $(v_i, w_i)$ is transitive in $D_i$ if and only if $v_i \to w_i$ is disimplicial in $G_{i}$.  Hence, by induction, $S = S_k$ is a disimplicial $M$-elimination of $G$ if and only if $T_k$ is a transitive $V$-elimination of $D$ for $V = V_1$ and $D = D_1$.  Moreover, $S$ is maximal if and only if $T_k$ is maximal.  This discussion is summarized in the following theorem.

\begin{theorem}\label{thm:transitive elimination}
  Let $M$ be a matching of a digraph $G$, $S = v_1 \to w_1, \ldots, v_k \to w_k$ be a sequence of arcs of $G$, $D = \JOIN(G, M)$, and $T = (v_1, w_1) , \ldots, (v_k, w_k)$.  Then, $S$ is a maximal disimplicial $M$-elimination of $G$ if and only if $T$ is a maximal transitive $V$-elimination of $D$. 
\end{theorem}

In view of Theorem~\ref{thm:transitive elimination}, we discuss how to obtain a maximal transitive $V$-elimination of a digraph $D_1 = D$.  The algorithm works in an iterative manner from $D_1 = D$.  At each step $i$, a transitive vertex $v_i \in V$ is removed from $D_i$ so as to obtain $D_{i+1}$; if no such vertex exists, then the algorithm halts with output $v_1, \ldots, v_{i-1}$.  To be able to find $v_i$ efficiently, the following data is maintained by the algorithm prior to the execution of iteration $i$:
\begin{itemize}
 \item $D_i$, implemented with the $h$-digraph structure,
 \item the set of transitive vertices $T_i$ of $D_i$,
 \item the number $t_i(v)$ of arcs leaving $N^-(v)$ and entering $N^+(v)$ in $D_i$, for $v \in V(D_i)$.
\end{itemize}
With the above information, any vertex of $T_i$ is taken by the algorithm to play the role of $v_i$. Once $v_i$ is selected, the algorithm has to update its data structure for the next iteration.  The update of $D_i$ into $D_{i+1}= D_i \setminus \{v_i\}$ is handled by the \texttt{Remove} operation of the $h$-digraph structure.  The update of $t_i$ into $t_{i+1}$ is done in two phases.  The first phase decrements $t_i(w)$ by $1$ for each arc $z \to w$ such that $w, z \in N^-(v)$, while the second phase decrements $t_i(w)$ by $1$ for each arc $w \to z$ such that $w,z \in N^+(v)$.  The \texttt{N'} operation of the $h$-digraph structure is employed for this step.  Finally, observe that $w \in T_{i+1}$ if and only if either $w \in T_i$ or $w \in N(v_i)$ and $t_{i+1}(w) = d^-(w)d^+(w)$.  Thus, the update of $T_i$ into $T_{i+1}$ takes $O(d(v_i))$ time.  Before the first step can take place, $D_1$ is initialized with an invocation to \texttt{Initialize}.  Note that \texttt{Remove} and \texttt{N'} are called $O(1)$ times for each vertex of $D$, thus $O(\alpha m)$ total time is consumed by the algorithm.  As for the space, $D_i$ requires $O(m)$ space while the remaining variables consume $O(n)$ bits.

Since $D = \JOIN(G, M)$ can be computed in linear time, $\alpha_G = \Theta(\alpha_D)$, and $m_G = \Theta(m_D)$ we conclude that a maximal disimplicial $M$-elimination can be computed in $O(\alpha_G m_G)$ time and linear space.

\section{Reduced dicliques}
\label{sec:WDI and DI}

By definition, a reflexive vertex $v$ is transitive if and only if $v \to v$ is disimplicial. Hence, if $D$ is an order graph, then $E(D)$ can be partitioned into a family of dicliques, all of which are reduced.  Moreover, by Proposition~\ref{prop:SPLIT dicliques}, $G = \SPLIT(D)$ is an ST graph and $E(G)$ can also be partitioned into a family of dicliques, all of which are reduced.  The purpose of this section is to study two graph classes that admit this kind of partition.

\subsection{Weakly diclique irreducible digraphs}

Say that a digraph is \emph{weakly diclique irreducible (WDI)} when all its arcs belong to a reduced diclique.  By Propositions \ref{prop:SPLIT dicliques}~and~\ref{prop:twin reduction}, $G$ is WDI if and only both $\SPLIT(G)$ and $\REPG(G)$ are WDI; for this reason, we consider only ST graphs with no twins for this section.  The next theorem, combined with Lemma~\ref{lem:bip then join}, shows that there is a one-to-one correspondence between the class of twin-free ST graphs that admit a perfect matching of disimplicial arcs and the class of order graphs.  A direct consequence of this theorem is that the recognition of WDI digraphs is harder than the recognition of order graphs.

\begin{theorem}
 A reflexive oriented graph $D$ is transitive if and only if $G = \SPLIT(D)$ is WDI.  Furthermore, if $G$ is WDI, then the perfect matching $M = \{\OUT(v) \to \IN(v) \mid v \in V(D)\}$ is the set of disimplicial arc of $G$.
\end{theorem}

\begin{proof}
 If $D$ is a reflexive oriented graph, then (i) $M$ is a perfect matching of $G$, and (ii) $\OUT(v) \to \IN(w)$ and $\OUT(w) \to \IN(v)$ are both arcs of $G$ if only if $v = w$.  Then, $\OUT(v) \to \IN(v)$ belongs to a reduced diclique if and only if it is disimplicial.  Since every arc $\OUT(v) \to \IN(w) \in E(G)$ belongs to the diclique $\{\OUT(v)\} \to \{\IN(v), \IN(w)\}$ of $G$, we conclude that $G$ is WDI if and only if all the arcs of $M$ are disimplicial.  Therefore, by Theorem~\ref{thm:transitive is disimplicial}, $G$ is WDI if and only if $D$ is transitive.  Moreover, since $G$ is twin-free by (ii), and the set of thin arcs is a matching containing $M$ by Lemma~\ref{lem:disimplicial is thin}, we conclude that no arc of $E(G) \setminus M$ is disimplicial.
\end{proof}

The following theorem shows that the recognition of WDI digraphs is simpler than the problem of listing the \emph{acyclic triangles} $a \to b$, $a \to c$, $c \to b$ of a digraph.  All such triangles can be found in either $O(\alpha m)$ time and $O(m)$ space or $O(n^\omega)$ time and $\Theta(n^2)$ space~\cite{ChibaNishizekiSJC1985}.  We conclude then that, unless it is proved that recognizing order graphs is strictly easier than listing triangles, the recognition of WDI digraphs is well solved. 

\begin{theorem}\label{thm:WDI characterization}
  An ST graph $G$ with no twins is WDI if and only if:
  \begin{itemize}
   \item $D = \JOIN(G)$ is transitive, and 
   \item for every arc $a \to b$ of $D$ there exists a vertex $c$ of $D$ such that $a \to c$ and $c \to b$ are also arcs of $D$.   
  \end{itemize}
\end{theorem}

\begin{proof}
  Suppose $G$ is WDI.  By definition, every vertex $(v, w)$ of $D$ that is neither a source nor a sink corresponds to a thin arc $v \to w$ of $G$.  Since $G$ is WDI, we know that $v \to w$ belongs to a diclique $N(y) \to N(x)$ for some disimplicial arc $x \to y$, thus $N(x) \subseteq N(v)$ and $N(y) \subseteq N(w)$. Moreover, taking into account that $v \to w$ is thin, it follows that $d(v) \leq d(x)$ and $d(w) \leq d(y)$, thus $N(v) = N(x)$ and $N(w) = N(y)$.  Therefore, $v \to w$ is disimplicial in $G$ and, by Theorem~\ref{thm:disimplicial is transitive}, $(v,w)$ is transitive in $D$; in other words $D$ is transitive.  Now, consider any arc $(v, v') \to (w',w)$ of $D$.  By definition, $v \to w$ is an arc of $G$ that belongs to some reduced diclique $N(y) \to N(x)$.  By Lemma~\ref{lem:disimplicial is thin}, $x \to y$ is a thin arc and, since $v \to y$ and $x \to w$, it follows that $(v, v') \to (x, y)$ and $(x,y) \to (w',w)$ are arcs of $D$.  

  For the converse, let $v \to w$ be any arc of $G$ and $(v, v')$ and $(w', w)$ be the vertices of $D$ that correspond to $v$ and $w$ (possibly $v = w'$).  By definition, $(v,v') \to (w,w')$ is an arc of $D$, thus, there exists a vertex $(x, y)$ of $H$ such that $(v,v') \to (x,y)$ and $(x, y) \to (w',w)$ are arcs of $D$ (possibly $v = x$ or $y = w$).  Since $(x,y)$ is neither a source nor a sink of $D$, then it follows that $(x,y)$ is transitive in $D$ and $x \to y$ is a thin arc of $G$.  So, by Theorem~\ref{thm:disimplicial is transitive}, $x \to y$ is a disimplicial arc of $G$ which means that $N(y) \to N(x)$ is a reduced diclique.  Now, taking into account that $(v,v') \to (x,y)$ and $(x, y) \to (w',w)$ are arcs of $D$, it follows that $v \in N(y)$ and $w \in N(x)$, i.e., $v \to w$ belongs to a reduced diclique.  In other words, $G$ is WDI.
\end{proof}

\subsection{Diclique irreducible digraphs}

In the remaining of this section we work with a subclass of WDI graphs, namely the diclique irreducible digraphs.  A digraph $G$ is \emph{diclique irreducible (DI)} when all its maximal dicliques are reduced.  Again, $G$ is DI if and only if both $\SPLIT(G)$ and $\REPG(G)$ are DI, thus we restrict our attention to ST graphs with no twins.  By Theorem~\ref{thm:WDI characterization}, we know that $\JOIN(G)$ is a transitive oriented graph; the following lemma proves that $\JOIN(G)$ must also be reflexive.

\begin{lemma}
 If an ST graph with no twins is DI, then its set of thin arcs is a perfect matching.
\end{lemma}

\begin{proof}
  Let $G$ be an ST graph that is DI and has no twins, $v$ be a source vertex of $G$, and $d(w)$ be minimum among the neighbors of $v$.  Since $G$ is DI, it follows that $v \to w$ belongs to some diclique $N(y) \to N(x)$ for a disimplicial arc $x \to y$.  Then $N(w) = N(y)$ which implies that $w = y$ as $G$ is twin-free.  Consequently, the thin neighbor of $v$ is $\TN(v) = w$.  Moreover, as $x \to y = w$ is disimplicial, it follows that the thin neighbor of $w$ is $\TN(w) = x$.  Suppose, to obtain a contradiction, that $x \neq v$.  Then, since $d(x) < d(v)$, we conclude that there exists $z \in N(v) \setminus N(x)$.  Thus, $\{v\} \to \{w,z\}$ is a diclique that must be contained in $B = N(b) \to N(a)$ for some disimplicial arc $a \to b$.  The same arguments used before allow us to conclude that $w = b = \TN(v)$ and $\TN(w) = a$.  This is clearly a contradiction because $x = \TN(w)$ does not belong to $B$ as it is not adjacent to $z$. We conclude, therefore, that $v = \TN(w) = \TN(\TN(v))$.  Analogously, $w = \TN(\TN(w))$ for every sink vertex $w$, thus every vertex belongs to a thin arc.  That is, the set of thin arcs is a perfect matching of $G$.
\end{proof}

\begin{corollary}
 If an ST graph with no twins is DI, then\/ $\JOIN(G)$ is an order graph.
\end{corollary}

Recall that order graphs are the graph theoretical equivalents of finite posets.  When $G$ is DI, the poset defined by $\JOIN(G)$ turns out to be what in order theory is known under the name of \emph{dedekind complete}.  We do not define what a dedekind complete poset is; in turn, we translate this concept in graph theoretic terms.

Let $D$ be a digraph.  Say that $u \in V(D)$ (resp.\ $\ell \in V(D)$) is an \emph{upper bound} (resp.\ a \emph{lower bound}) of $V \subseteq V(D)$ when $v \to u$ (resp.\ $\ell \to v$) for every $v \in V$.  We write $\mu(V)$ and $\lambda(V)$ to denote the sets of upper and lower bounds of $V$, respectively.  When $\mu(V)$ (resp.\ $\lambda(V)$) is nonempty, the set $V$ is said to be \emph{bounded from above (resp.\ below)}.  Every lower bound of $\mu(V)$ that belongs to $\mu(V)$ is a \emph{supremum} of $V$, while every upper bound of $\lambda(V)$ that belongs to $\lambda(V)$ is an \emph{infimum of $V$}.  Note that $V$ has at most one supremum (resp.\ infimum) when $D$ is an oriented graph.  A \emph{dedekind graph} is an order graph $D$ such that every $\emptyset \subset V \subseteq V(D)$ that is bounded from above has a supremum.  It is well known that an order graph $D$ is dedekind if and only if every $\emptyset \subset V \subseteq V(D)$ that is bounded from below has an infimum.

The reason why dedekind graphs come into play in the characterization of DI graphs has to do with the way $\JOIN(G)$ encodes the dicliques and disimplicial arcs of $G$.  Roughly speaking, a disimplicial arc $v \to w$ of $G$ is a transitive vertex $(v,w)$ of $\JOIN(G)$ where $N(v)$ and $N(w)$ corresponds to the lower and upper bounds $L, U$ of $\{(v,w)\}$, respectively.  Moreover, $(v,w)$ is both the infimum and supremum of $U$ and $L$, respectively.  This somehow explains why dedekind graphs appear when every diclique has a disimplicial arc.  The complete proof is given in the next theorem.  

\begin{lemma}\label{lem:dicliques in JOIN}
  Let $G$ be a digraph, $V, W$ be nonempty subsets of $V(G)$, $D = \JOIN(G)$, and $L = \{(v,v') \in V(D) \mid v \in V\}$ and $U = \{(w',w) \in V(D) \mid w \in W\}$.  Then,  $V \to W$ is a diclique of $G$ if and only if $L \subseteq \lambda(U)$ and $U \subseteq \mu(L)$.  Furthermore, $V \to W$ is a maximal diclique exactly when $L = \lambda(U)$ and $U = \mu(L)$.
\end{lemma}

\begin{proof}
  Just observe that, by definition, $v \to w \in E(G)$ for every $v \in V$ and $w \in W$ if and only if $(v, v') \to (w',w) \in E(D)$ for every $(v,v') \in L$ and $(w', w) \in U$.  That is, $V \to W$ is a diclique of $G$ if and only if $L \subseteq \lambda(U)$ and $U \subseteq \mu(L)$.  Moreover, using the same argument, the maximality of $V \to W$ occurs precisely when $L = \lambda(U)$ and $U = \mu(L)$.
\end{proof}

\begin{theorem}\label{thm:DI characterization}
 Let $G$ be an ST graph with no twins.  Then $G$ is DI if and only if\/ $\JOIN(G)$ is dedekind.
\end{theorem}

\begin{proof}
  Suppose $G$ is DI, let $D = \JOIN(G)$, and consider any nonempty $M \subseteq V(D)$ bounded from above.  Let (a) $U = \mu(M)$ and (b) $L = \lambda(U)$, and observe that (c) $U = \mu(L)$.  By definition, $L = \{(v,v') \in V(D) \mid v \in V\}$ and $U = \{(w',w) \in V(D) \mid w \in W\}$ for some  $V, W \subseteq V(G)$.  By Lemma~\ref{lem:dicliques in JOIN}, $V \to W$ is a maximal diclique of $G$, thus it contains some disimplicial arc $v \to w$.  By Lemma~\ref{lem:disimplicial is thin}, $v \to w$ is a thin arc, thus $(v,w)$ is a vertex of $D$.  Moreover, $(v,w) \in L \cap U$ because $v \in V$ and $w \in W$.  Then, by (b) and (c), it follows that $(v,w)$ is the supremum of $L$ and the infimum of $U$, while by (a), $(v,w)$ is a supremum of $M$ as well.
 
  For the converse, suppose $V \to W$ is a maximal diclique of $G$ and let (a) $L = \{(v,v') \in V(D) \mid v \in V\}$ and (b) $U = \{(w',w) \in V(D) \mid w \in W\}$.  By Lemma~\ref{lem:dicliques in JOIN}, $L = \lambda(U)$ and $U = \mu(L)$, hence, since $D$ is dedekind, it follows that $L \cap U$ contains some vertex $(v,w)$ such that (c) $L = N^-((v,w))$ and (d) $U = N^+((v,w))$.  By (a) and (c), and considering how $\JOIN$ works, we conclude that $V = N_G(w)$, while $W = N_G(v)$ by (a) and (d).  In other words, $v \to w$ is a disimplicial arc of $V \to W$. 
\end{proof}

\begin{corollary}\label{cor:dedekind characterization}
  A digraph $D$ is dedekind if and only if\/ $\SPLIT(D)$ is DI.
\end{corollary}

\begin{proof}
  By Lemma~\ref{lem:bip then join}, $D = \JOIN(G)$ for $G = \SPLIT(D)$, while, by Theorem~\ref{thm:DI characterization}, $D$ is dedekind if and only if $G$ is DI.  
\end{proof}

By Theorem~\ref{thm:DI characterization} and Corollary~\ref{cor:dedekind characterization}, DI and dedekind graphs are equally hard to recognize, and the recognition can be done in polynomial time rather easily.  Just observe that a DI graph has at most $m$ maximal dicliques, one for each disimplicial arc.  Then, a recognition algorithm needs to traverse at most $m+1$ maximal dicliques before finding one that is not reduced.  To test if a diclique is reduced, it is enough to check that it contains a precomputed disimplicial arc.  Since the disimplicial arcs can be found in $O(\alpha m)$ time, and the $m+1$ dicliques of can be traversed in $O(nm^2)$ time~\cite{DiasFigueiredoSzwarcfiterDAM2007}, an $O(nm^2)$ time algorithm is obtained.  We now describe an $O(nm)$ time and $O(m)$ space algorithm that exploits the definition of dedekind graphs.  The following simple lemma is the key of the algorithm.

\begin{lemma}
  An order graph $D$ is dedekind if and only if for every $v,w \in V(G)$ with $\mu(\{v,w\}) \neq \emptyset$ there exists $u \in V(G)$ such that $|\mu(\{v,w\})| = d^+(u)$.
\end{lemma}

\begin{proof}
  Suppose $D$ is dedekind and let $u$ be the supremum of $\{v,w\}$, for $\{v,w\} \subseteq V(D)$ bounded from above.  By definition, $v \to u \in V(D)$ and $w \to u \in V(D)$, thus $N^+(u) \subseteq \mu(\{v,w\})$ because $u$ is transitive.  Also by definition, $u \to z$ for every $z \in \mu(\{v,w\})$, thus $\mu(\{v,w\}) \subseteq N^+(u)$.  Therefore, $|\mu(\{v,w\})| = |N^+(u)| = d^+(u)$.

  For the converse, observe again that $N^+(u) \subseteq \mu(\{v,w\})$ for every $u \in \mu(\{v,w\})$, because $u$ is transitive.  So, if $u \in \mu(\{v,w\})$ has degree $|\mu(\{v,w\})|$, then $N^+(u) = \mu(\{v,w\})$, which means that $\{v,w\}$ has a supremum.  That is, $\{v,w\}$ has a supremum for every $\{v,w\} \subseteq V(D)$ bounded from above.  It is well known (taking into account that dedekind graphs correspond to dedekind complete finite posets) that, in this case, $D$ is dedekind.
\end{proof}

The algorithm to determine if an order digraph $D$ is dedekind traverses $\mu(\{v,w\})$, for each pair of vertices $v,w \in V(G)$, searching for a vertex $u$ with $d^+(u) = |\mu(v,w)|$.  For the implementation, an outer loop traverses each $v \in V(G)$ and an inner loop traverses each $w \in V(G) \setminus \{v\}$.  Before the inner loop begins, all the vertices in $N^+(v)$ are marked in $O(d(v))$ time.  Then, in the inner loop, $\mu(\{v,w\})$ is obtained in $O(d(w))$ time by filtering those vertices of $N^+(w)$ that are marked.  The degree of all the vertices in $\mu(\{v,w\})$ is the evaluated in $O(d(w))$ time as well.  The total time required by the algorithm is, therefore, \[O\left(\sum_{v \in V(G)}\left(d(v) + \sum_{w \in V(G)}d(w)\right)\right) = O(nm),\]
while the space complexity is $O(m)$ bits.  Since order graphs can be recognized in $O(\alpha m)$ time and $O(m)$ space, we conclude that the recognition DI and dedekind graphs takes $O(nm)$ time and $O(m)$ space.

\section{Results on bipartite graphs and further remarks}
\label{sec:further remarks}

A \emph{bipartite graph} is a triple $G = (V, W, E)$ where an unordered pair $vw$ belongs to $E$ only if $v \in V$ and $w \in W$.  An edge $vw$ is \emph{bisimplicial} when every vertex in $N(v)$ is adjacent to all the vertices in $N(w)$.  By replacing each $vw$ by an arc $v \to w$, an ST graph $\vec{G}$ is obtained.  Moreover, an edge $vw$ of $G$ is bisimplicial precisely when $v \to w$ is disimplicial in $\vec{G}$.  So, the algorithms in this article can be applied directly to bipartite graphs so as to solve the corresponding problems.  In this section we summarize the results for bipartite graphs while we provide further remarks.

In Section~\ref{sec:disimplicial vs transitive} we proved that listing the bisimplicial edges of a bipartite graph and finding the transitive vertices of a digraph are equally hard problems.  The good news is that the bisimplicial edges of a bipartite graph can be found in $O(\alpha m)$ time, improving over the previous $O(nm)$ time algorithm; the bad news is that we cannot improve this algorithm further using only $O(m)$ space, unless an $o(\alpha m)$ time algorithm to find the transitive vertices of a digraph is provided.  

In Section~\ref{sec:disimplicial elimination} we describe an $O(\min\{\Delta \eta, m\}m)$ time and $O(m)$ space algorithm to compute a maximal disimplicial elimination of $\vec{G}$.  When applied to bipartite graphs, a maximal elimination scheme $S$ is obtained.  Since $\eta < \Delta$, our algorithm improves the worst-case time bound of~\cite{Bomhoff2011} for all the bipartite graphs with $\Delta = o(\sqrt{m})$.  Golumbic and Goss~\cite{GolumbicGossJGT1978} proved that $S$ is perfect whenever $G$ admits a perfect elimination scheme, thus the algorithm can be used to recognize if a sparse graph is perfect elimination bipartite.  The concept of perfect elimination graphs can be generalized to digraphs and disimplicial eliminations.  Just say that a digraph $D$ is \emph{perfect disimplicial elimination} whenever it admits a disimplicial elimination $S$ such that $G \setminus V(S)$ has no arcs.  Unfortunately, finding a maximal disimplicial elimination is not enough to determine if $D$ is perfect, as it is shown in Figure~\ref{fig:maximal not perfect}.  So, the recognition of perfect disimplicial elimination remains open.

\begin{figure}
  \centering
  \includegraphics{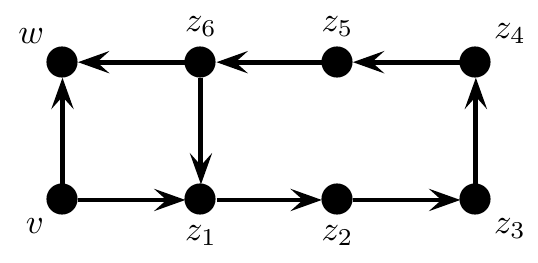}
  \caption{A perfect disimplicial elimination digraph with a non-perfect maximal disimplicial elimination: $v \to w$ and $v \to z_1, z_2 \to z_3, z_4 \to z_5, z_6 \to w$ are maximal disimplicial eliminations.}\label{fig:maximal not perfect}
\end{figure}

In Section~\ref{sec:disimplicial elimination} we also consider the problem of computing a maximal disimplicial $M$-elimination, for an input matching $M$, for which we provide an $O(\alpha m)$ time and $O(m)$ space algorithm.  Rose and Tarjan~\cite{RoseTarjanSJAM1978} proved that this problem is harder than determining if a given digraph is transitive.  Up to this date, the best algorithm to determine if a sparse graph is transitive costs $O(\alpha m)$ time and $O(m)$ space.  So, the problem is well solved, without using more than $O(m)$ space, unless better algorithms for recognizing transitive digraphs are found.

Recall one of the motivations for finding a maximal disimplicial elimination is to be able to perform some iterations of the Gaussian elimination process on a sparse matrix $M$ with the guaranty that no zero entry will change into a non-zero value.  Being $M$ sparse, we expect $\alpha_G \approx 1$ and $\Delta_G \approx 1$ for $G = G(M)$.  If so, then finding the disimplicial elimination and applying the corresponding iterations of the Gaussian elimination require linear time.  That is, our algorithm can be used to preprocess $M$, say before solving the system $Mx = b$.  In the worst case no zero fill-in entry is found and thus $M$ remains the same.  Yet, the extra time paid for this examination is low.

In Section~\ref{sec:WDI and DI} we deal with the classes of WDI and DI digraphs.  We noted that every order graph $D$ is uniquely associated with a twin-free ST graph $G$ that is WDI, namely $G = \SPLIT(D)$.  In fact, each $v \in V(D)$ gets transformed into the disimplicial arc $\OUT(v) \to \IN(v)$ of $G$, thus $G$ has a perfect matching of disimplicial arcs.  The converse is also true, any ST graph that has a perfect matching of disimplicial arcs must be isomorphic to $\SPLIT(D)$ for some order graph $D$.  We remark that the order relation $\to$ of $D$ is somehow preserved in $G$.  Indeed, note that $v \to w \in E(D)$ only if $w \to v \not\in E(D)$, thus $\OUT(v) \to \IN(w) \in E(G)$ while $\OUT(w) \to \IN(v) \not\in E(G)$.  Hence, by transitivity, $v \to w \in E(D)$ if and only if $N(\OUT(v)) \subset N(\OUT(w))$ and $N(\IN(w)) \subset N(\IN(v))$.  In this section we also proved that $G$ is also DI whenever $D$ is a dedekind graph.  Moreover, each $A \subseteq V(D)$ with supremum $u$ is associated with a reduced biclique $V \to W$ such that $V = \{\OUT(v) \mid v \to u \in E(D)\}$ and $W = \{\IN(w) \mid u \to w \in E(D)\}$.   Note that, in particular, $\OUT(u) \to \IN(u)$ is the disimplicial arc of $V \to W$.  

\newcommand{\doi}[1]{doi: \href{http://dx.doi.org/#1}{\dodoi{#1}}}
\newcommand{\dodoi}{\begingroup \urlstyle{rm}\Url}

% \bibliographystyle{notabbrvnat.bst}
% \bibliography{../../biblio/biblio}
% 
% \end{document}

\end{document}